\documentclass[10pt, conference, letterpaper]{IEEEtran}
\IEEEoverridecommandlockouts
\usepackage{cite}
\usepackage{amsmath,amssymb,amsfonts}
\usepackage{amsthm}
\newtheorem{theorem}{Theorem}

\DeclareMathOperator*{\argmax}{argmax}
\usepackage[ruled,vlined,linesnumbered]{algorithm2e}
\usepackage{graphicx}
\usepackage{textcomp}
\usepackage{subfigure}
\usepackage{xcolor}
\usepackage{url}
\begin{document}

\title{Efficient Serverless Function Scheduling at the Network Edge\\
}

\author{
\IEEEauthorblockN{
Jiong Lou\IEEEauthorrefmark{1},
Zhiqing Tang\IEEEauthorrefmark{2},
Shijing Yuan\IEEEauthorrefmark{1},
Jie Li\IEEEauthorrefmark{1},
and Chengtao Wu\IEEEauthorrefmark{1}}
\IEEEauthorblockA{\IEEEauthorrefmark{1}Department of Computer Science and Engineering, Shanghai Jiao Tong University, China}
\IEEEauthorblockA{\IEEEauthorrefmark{2}Guangdong Key Lab of AI and Multi-Modal Data Processing, BNU-HKBU United International College, China}
\IEEEauthorblockA{lj1994@sjtu.edu.cn, zhiqingtang@bnu.edu.cn,\{2019ysj,lijiecs,wuct\}@sjtu.edu.cn}
\thanks{\textit{(Corresponding author: Jie Li)}}
}

\maketitle

\begin{abstract}

Serverless computing is a promising approach for edge computing since its inherent features, e.g., lightweight virtualization, rapid scalability, and economic efficiency. However, previous studies have not studied well the issues of significant cold start latency and highly dynamic workloads in serverless function scheduling, which are exacerbated at the resource-limited network edge. In this paper, we formulate the Serverless Function Scheduling (SFS) problem for resource-limited edge computing, aiming to minimize the average response time. To efficiently solve this intractable scheduling problem, we first consider a simplified offline form of the problem and design a polynomial-time optimal scheduling algorithm based on each function's weight. Furthermore, we propose an Enhanced Shortest Function First (ESFF) algorithm, in which the function weight represents the scheduling urgency. To avoid frequent cold starts, ESFF selectively decides the initialization of new function instances when receiving requests. To deal with dynamic workloads, ESFF judiciously replaces serverless functions based on the function weight at the completion time of requests. Extensive simulations based on real-world serverless request traces are conducted, and the results show that ESFF consistently and substantially outperforms existing baselines under different settings.
\end{abstract}


\section{Introduction}
In edge computing, computation resources are deployed at the network edge to offer low-latency services \cite{shi2016edge}. The edge servers are still resource-limited compared with cloud servers. Besides, the mobility of users in edge computing generates highly dynamic workloads, which can be rapidly reduced to zero \cite{ouyang2018follow}. Serverless computing \cite{castro2019rise}, a new paradigm designed for dynamic and infrequent short-lived computations, is promising for edge computing \cite{aslanpour2021serverless}. Firstly, the serverless function is initialized on demand to avoid resource wastage due to pre-provision and long-running containers. Secondly, it can automatically and instantly scale up for peak workloads and scale down for zero workloads. Thirdly, it is more economical for it only charges for CPU time spent in request execution but does not charge for idle cycles \cite{wang2018peeking}.


Most of the current serverless platforms, e.g., OpenWhisk \cite{openwhisk}, AWS lambda \cite{awslambda} and Azure \cite{Azure}, are designed for cloud computing, unsuitable to resource-limited edge computing. They scale up serverless function instances when there is no idle instance for waiting requests \cite{wang2018peeking}. Though the waiting time for an idle instance may be eliminated in this way, this mechanism can easily bring function instance over-provision \cite{gunasekaran2020fifer} and frequent cold starts \cite{wu2022container} (i.e., the process of initializing a new function instance). These two problems result in delayed response and largely affect the Quality of Experience (QoE). Therefore, serverless computing at the network edge should be improved to reduce the average response time.

There are a few studies proposed to extend serverless computing to the network edge \cite{pan2022retention,hall2019execution,das2020performance,cicconetti2021faas,wang2021lass,bermbach2022auctionwhisk}. Most of these studies focus on data communication \cite{cicconetti2021faas}, function placement \cite{hall2019execution,bermbach2022auctionwhisk,das2020performance} and function cache \cite{pan2022retention}. However, all these studies still ignore the severe problem of serverless function over-provision in edge computing. Two schedulers are designed to partially mitigate the problem of function over-provision and frequent container startups, considering the scheduling of a single function \cite{OpenWhiskscheduler,wang2021lass}. The OpenWhisk V2 is designed to separate the control flow and data flow \cite{OpenWhiskscheduler}.  LaSS \cite{wang2021lass} estimates the instance number for each function to meet the individual deadline based on queue theory but neglects the cold start latency. 

However, these studies still do not comprehensively overcome the following challenges: (1) Massive cold starts are incurred by frequent replacement of serverless functions, which prolongs the current request's response time and delays the future requests' executions. Initializing a function instance for each waiting request \cite{openwhisk,fuerst2021faascache} or setting a waiting threshold \cite{OpenWhiskscheduler} cannot globally reduce cold starts or optimize the average response time. (2) Short serverless function requests are blocked by long requests, which increases the average response time. Request bursts are very common for highly dynamic workloads in edge computing \cite{9796969}, making the request blocking worse. These challenges make the serverless function scheduling in edge computing much intractable and result in poor performance of existing methods.

\begin{figure*}[t]
    \setlength{\abovecaptionskip}{0cm} 
    \setlength{\belowcaptionskip}{0cm} 
    \centering
    \subfigure[OpenWhisk]{
    \centering
    \includegraphics[width=0.3\textwidth]{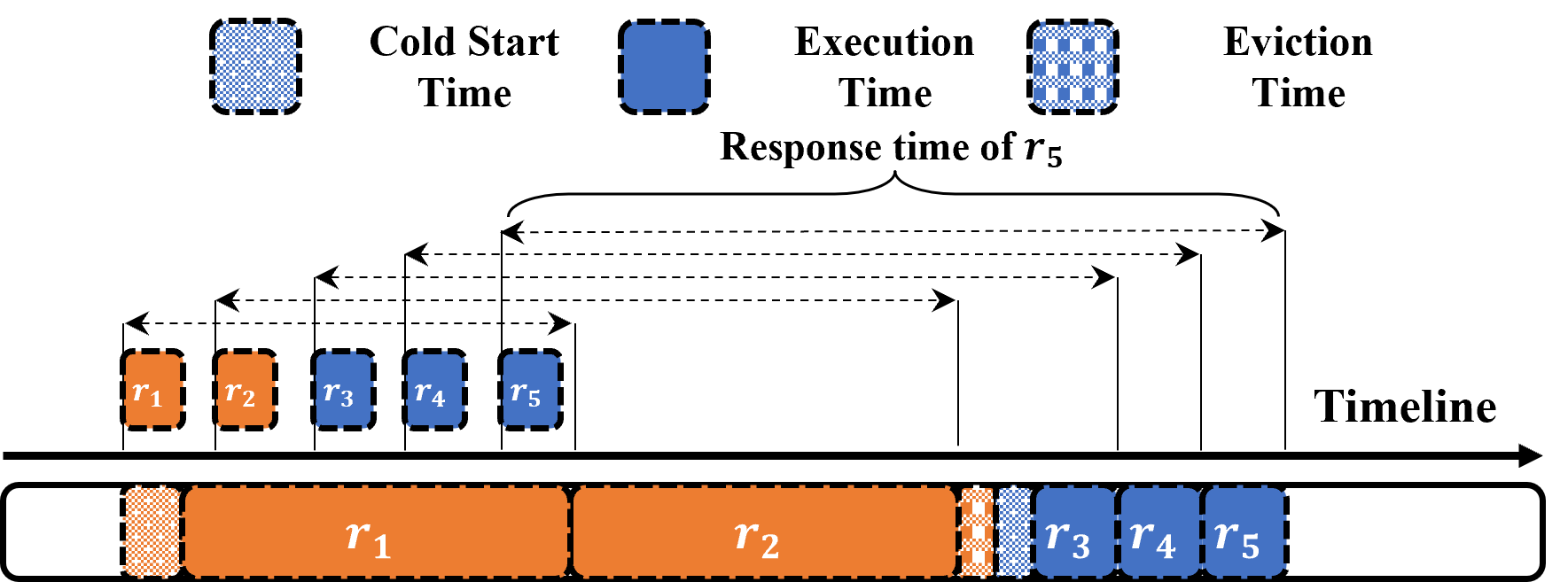}
    \label{OpenWhisk}
    }%
    \subfigure[OpenWhisk V2]{
    \centering
    \includegraphics[width=0.3\textwidth]{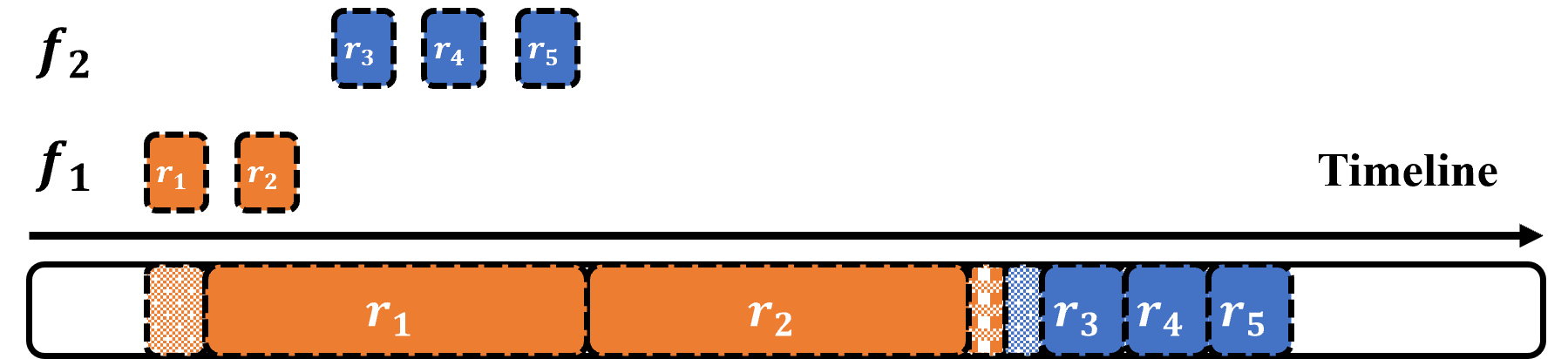}
    \label{OpenWhiskv2}
    }
    \subfigure[ESFF]{
    \centering
    \includegraphics[width=0.3\textwidth]{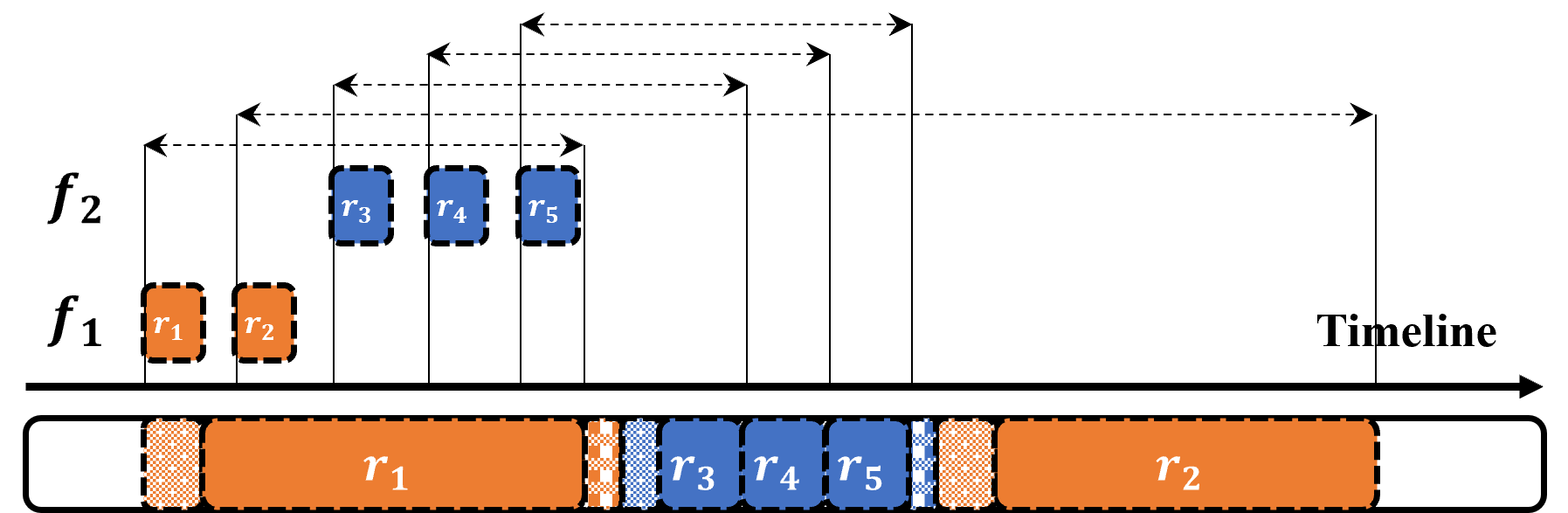}
    \label{ESFF}
    }%
    \centering
    \caption{\label{motivation} An illustrative scheduling example of OpenWhisk, OpenWhisk V2 and ESFF. The requests with the same color belong to the same function. Two long requests $r_1$, $r_2$ of function $f_1$ arrive earlier than there short requests $r_3$, $r_4$, $r_5$ of function $f_2$. ESFF achieves the shortest average response time. }
    \vspace{-0.3cm}
  \end{figure*}

In this paper, we formulate the Serverless Function Scheduling (SFS) problem in resource-limited edge computing, aiming to minimize the average response time, which is proved to be NP-hard. To efficiently solve this intractable scheduling problem, we first analyze a simplified form of the SFS (SSFS) problem: offline scheduling of requests on an edge server that can execute a single function at once. Then, based on the optimal scheduling algorithm of SSFS, we propose an Enhanced Shortest Function First (ESFF) algorithm, considering the aforementioned challenges. ESFF consists of two sub-policies: (1) Function Creation Policy (FCP). FCP selectively initializes a new function instance at the request arrival time by judging whether the average response time can benefit from the function initialization. (2) Function Replacement Policy (FRP). FRP replaces functions at the request completion time according to the function weight, i.e., the scheduling urgency of a function.

To the best of our knowledge, we are the first to formulate the Serverless Function Scheduling problem, fully considering the significant cold start latency and highly dynamic function workloads in resource-limited edge computing. The contributions of this paper are summarised as follows:
\begin{enumerate}
 \item Firstly, we analyze a simplified problem SSFS and prove that the optimal scheduling of SSFS can be obtained in polynomial time. 
 \item Secondly, we design a lightweight scheduling algorithm to solve the SFS problem in an online manner. ESFF judiciously prioritizes each serverless function based on the execution time, the cold start time, and the number of waiting requests. It initializes and replaces new function instances based on function weights. 
 \item Finally, extensive simulations based on real-world serverless request traces \cite{zhang2021faster} show that ESFF substantially and consistently outperforms existing baselines under different settings.
\end{enumerate}

\section{Motivation}
In Fig.~\ref{motivation}, we give a simple example to illustrate the motivation of serverless function scheduling. Fig.~\ref{OpenWhisk} and \ref{OpenWhiskv2} generated by different scheduling algorithms show the same results. In Fig.~\ref{ESFF}, the proposed algorithm ESFF produces the best scheduling in terms of the average response time.

In Fig.~\ref{OpenWhisk}, OpenWhisk processes the requests in the central queue based on the ascending order of their arrival time. As a result, short requests $r_3$, $r_4$ and $r_5$ are seriously blocked by long requests $r_1$ and $r_2$. In a real edge environment, the request bursts are very common \cite{varghese2016challenges}, so many short requests will be blocked by long requests with the scheduling of OpenWhisk. In Fig.~\ref{OpenWhiskv2}, OpenWhisk V2 maintains an individual queue for each function. If there are requests waiting in the queue, the required function instance will continue to process these requests. This design still results in blocking. In Fig.~\ref{OpenWhiskv2}, when request $r_1$'s execution is finished, request $r_2$ has already arrived at the edge server, waiting in the queue. Moreover, if multiple requests of the function $f_1$ arrive at the edge server before finishing the request $r_2$, the requests $r_3$, $r_4$, and $r_5$ will be blocked again, which is unreasonable and far from optimal.

In Fig.~\ref{ESFF}, ESFF also sends each request to its individual function queue. It differs from OpenWhisk V2 in that after finishing a request, a function instance can be selectively replaced by another function instead of processing the requests in its queue. Therefore, short requests will not be largely blocked by long ones, and ESFF achieves the best performance.


\section{System Model and Problem Formulation}

\subsection{System Model}
\textit{Request Model:} In this paper, the requests of different serverless functions are released by mobile users and arrive at the serverless platform over time. A request is denoted as $r_{i} \in \mathbf{R}$, where $\mathbf{R}$ is the set of requests. Request $r_{i}$'s arrival time is denoted as  $t^a_i$. The execution time of request $r_{i}$ is denoted as $t^e_{i}$. Due to the dynamic nature of edge computing, the arrival time and execution time of each request are unknown in prior and hard to predict \cite{8057116}. The start execution time and the completion time of the request $r_{i}$ are denoted as $t^{s}_{i}$ and $t^{c}_{i}$, respectively.



\textit{Serverless Function Model:} The function set is denoted by $\mathbf{F} = \{f_1,f_2, \dots, f_{|\mathbf{F}|}\}$. The function of request $r_{i}$ is denoted by $l_{i} \in \mathbf{F}$. Before processing a request, the corresponding function instance should be initialized, named cold start. For a function $f_{j}$, the cold start latency is defined as $t^l_{j}$. In this paper, it is assumed that request execution cannot be interrupted. $k^j_o$ represents the $o$-th instance of $f_j$, with two states: (1) Idle state, $state(k^j_o)=1$, when waiting for requests. (2) Busy state, $state(k^j_o)=1$, when processing a request. After finishing a request, the function instance turns idle, waiting for future requests. An idle function can be evicted to release the resources. For a function $f_{j}$, the eviction time is defined as $t^v_{j}$. The set of initialized instances of function $f_j$ is defined as $\mathbf{K}^j = \{k^j_1,k^j_2,\dots,k^j_{|\mathbf{K}^j|}\}$, which changes over time.

\textit{Edge Server Model:} For simplicity, the serverless platform is assumed to be deployed on a resource-constrained edge server\footnote{Multiple edge servers inter-connected by the ultra-low latency network can be modeled as a powerful edge server for the neglectable transmission time.}. The resource capacity of the edge server is represented by the maximum number of function instances that can be executed concurrently, denoted as $C$. 


\subsection{Problem Formulation}
For each request $r_i$, it cannot be processed before arriving, i.e., $t^a_i < t^s_i$. As each request's processing cannot be interrupted, therefore, $t^c_i = t^s_i+t^e_i$. Before processing a request $r_i$, an instance of the corresponding function $l_{i}$ should be initialized and in the idle state.

At any time, the edge server can hold at most $C$ function instances. Each function instance can process one request of its type at once~\cite{openwhisk}. When an instance of $f_j$ is initialized without evicting any function instance, it requires $t^l_j$ for cold start; When it is initialized by evicting an instance of function $f_{j^{\prime}}$, it requires $t^v_{j^{\prime}}$ for eviction and $t^l_j$ for cold start. 


The response time of a request $r_{i}$ is $t^r_i = t^c_i-t^a_i$, including both the execution time and waiting time. The objective is to optimize the average response time of all requests:
\begin{equation}
    \setlength{\belowdisplayskip}{2pt}
    \setlength{\abovedisplayskip}{2pt}
    T = \frac{1}{|\mathbf{R}|}\sum_{r_i \in \mathbf{R}}{t^r_i}.
\end{equation}
The SFS problem is an online scheduling problem, i.e., making scheduling decisions without future information.

\begin{theorem} \label{sfsnphard}
    The SFS problem is NP-hard.
  \end{theorem}


  \begin{proof}
  To prove the NP-hardness of the SFS problem, we consider the scheduling families (akin to functions in this paper) of jobs on parallel machines involving sequence-dependent setup time \cite{liaee1997scheduling}. The scheduling problem is a special case of the SFS problem: the eviction time of each function is set to zero, all requests arrive at the edge server at the time of zero, and all information of requests is known before scheduling. The NP-hardness of this scheduling problem has been proved in \cite{liaee1997scheduling}. Therefore, the SFS problem is NP-hard.
 \end{proof}


\section{Optimal Scheduling Algorithm of SSFS}
Due to the NP-hardness of SFS, we consider a simplified form of SFS (SSFS) and design an optimal scheduling algorithm, which can efficiently be extended for the SFS problem.

SSFS is simplified in the following aspects: (1) The edge server is unary, i.e., holding at most one initialized function instance at any time. (2) The requests of the same function have identical execution time, i.e., $t^e_i = t^e_{i^{\prime}}=t_j$ if $l_i=l_{i^{\prime}}=f_j$. $t_j$ is the request $f_j$'s execution time and also named as $f_j$'s execution time. (3) All requests arrive at the edge server at time 0. (4) All functions' request number, cold start time, eviction time, and execution time are known before scheduling. The request number of the function $f_j$ is denoted as $n_j$.

\textbf{Optimal scheduling algorithm of the SSFS problem:} Since the edge server is unary, the scheduling of SSFS is equivalent to sequencing requests. Similar to the shortest job first, the optimal scheduling algorithm of SSFS has two steps. (1) Each function is associated with a weight, which is computed by $w_j = t_j + \frac{t^l_j+t^v_j}{n_j}$. (2) Each function's requests are processed continuously and processed in ascending order of function weights.


\begin{theorem} \label{SSFSoptimal}
    The above algorithm generates the optimal scheduling of the SSFS problem.
  \end{theorem}


\begin{figure}[tbp]
    \begin{center}
      \includegraphics[width=0.9\linewidth]{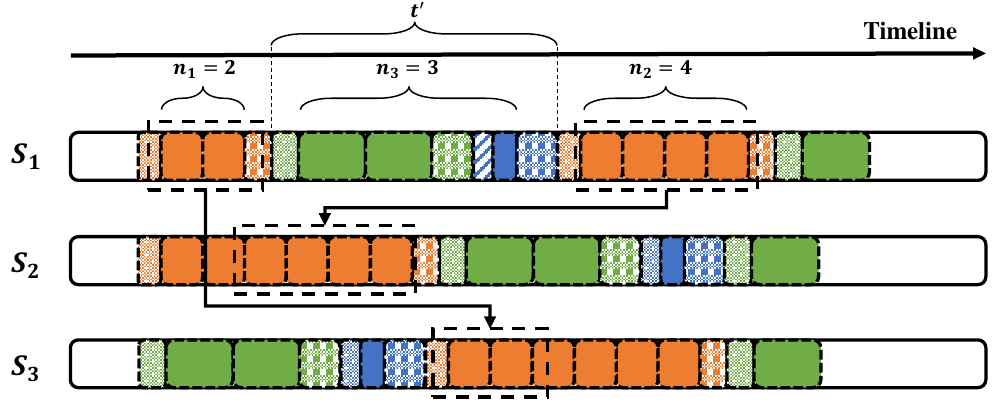}
      \caption{\label{putting1} The schedule $S_1$ is changed to $S_2$ and $S_3$ by putting requests of the same function together.}
  \end{center}
  \end{figure}
\addtolength{\topmargin}{0.041in}

  \begin{proof}
  The proof has two parts. We first prove that the requests of the same function should be processed continuously and then show that the function processing sequence follows the ascending order of function weights.
  
 We prove the first part by contradiction. As shown in Fig.~\ref{putting1}, it is assumed that in an optimal schedule $S_1$, multiple requests of the same function $f_j$ are divided into two parts by some other functions' requests. The request numbers of the two parts are denoted as $n_1$ and $n_2$, respectively. The request number of other functions' requests is $n_3$. These $n_3$ requests consist of requests different functions and the total time duration of $n_3$ requests is $t^{\prime}$ (including their execution time, eviction time, and cold start time). The total response time of this schedule is $T_1$, the total response time of putting $n_2$ requests before $n_3$ requests (the schedule is named $S_2$) is $T_2$, and the total response time of putting $n_1$ requests after $n_3$ requests (the schedule is named $S_3$) is $T_3$. The putting operations are shown in Fig.~\ref{putting1}, where $n_1=2$, $n_2=4$ and $n_3=3$. The number of requests processed after all these requests is denoted as $n_4$. We compare the $T_1$ and $T_2$:
 \begin{equation}
     T_1-T_2 = n_2(t^l_i+t^v_i+t^{\prime})-n_3n_2t^e_i+(t^l_i+t^v_i)n_4. \notag
 \end{equation}
 Since $T_1$ is the minimum, $0 \geq T_1-T_2$ and:
 \begin{align}
      0 & \geq n_2(t^l_i+t^v_i+t^{\prime})-n_3n_2t^e_i+(t^l_i+t^v_i)_in_4 \notag\\&> n_2t^{\prime}-n_3n_2t^e_i \notag\\&= n_2(t^{\prime}-n_3t^e_i).\notag
 \end{align}
 We obtain:
 \begin{equation} \label{contradtion1}
       n_3t^e_i > t^{\prime}.
 \end{equation}
 Then, $T_1$ and $T_3$ are compared:
 \begin{equation}
     T_1-T_3 = n_3(t^l_i+t^v_i+t^e_in_1)+(n_4+n_2)(t^l_i+t^v_i)-t^{\prime}n_1. \notag
 \end{equation}
 Since $0 \geq T_1-T_3$, we obtain:
 \begin{align}
 0 &\geq n_3(t^l_i+t^v_i+t^e_in_1)+(n_4+n_2)(t^l_i+t^v_i)-t^{\prime}n_1 \notag\\ &> n_1n_3t^e_i - n_1t^{\prime} \notag\\&= n_1(n_3t^e_i-t^{\prime}). \notag
 \end{align}
 Therefore,
 \begin{equation} \label{contradtion2}
     t^{\prime} > n_3t^e_i.
 \end{equation}
  There exists a contradiction in Eqs.\eqref{contradtion1} and \eqref{contradtion2}. One of $T_2$ and $T_3$ is less than $T_1$. Therefore, in the optimal schedule of SSFS, requests of the same functions should be processed continuously.
  
  \begin{figure}[tbp]
    \begin{center}
      \includegraphics[width=0.9\linewidth]{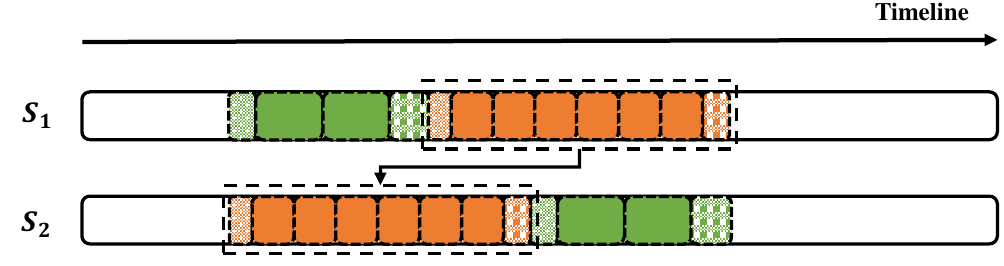}
      \caption{\label{putting2} The schedule $S_1$ is changed to $S_2$ by exchanging the execution order of two functions' requests.}
  \end{center}
  \end{figure}
  Then, we prove the second part that functions' executions are sorted in the ascending order of their weights. In the optimal schedule of SSFS, we assume that the requests of the function $f_j$ with a lower weight are processed after $f_{j^{\prime}}$ with a higher weight. The schedule is denoted as $S_1$, and the total response time is denoted as $T_1$. Then, we put the function $f_j$'s executions before the function $f_{j^{\prime}}$. The changed schedule is denoted as $S_2$, and its total response time is denoted as $T_2$. Fig.~\ref{putting2} shows how to change $S_1$ to $S_2$. $T_1$ and $T_2$ are compared:
  \begin{align}
      T_1 - T_2 = &n_j(n_{j^{\prime}}t^e_{j^{\prime}}+t^v_{j^{\prime}}+t^l_{j^{\prime}}) - n_{j^{\prime}}(n_jt^e_j+t^v_{j}+t^l_{j}) \notag \\= &n_jn_{j^{\prime}}(t^e_{j^{\prime}} + \frac{t^v_{j^{\prime}}+t^l_{j^{\prime}}}{n_{j^{\prime}}}-t^e_j-\frac{t^v_{j}+t^l_{j}}{n_i}). \label{t1t2compare}
  \end{align}
  Since we have assumed that $w_j < w_{j^{\prime}}$, then $0 < t^e_{j^{\prime}} + \frac{t^v_{j^{\prime}}+t^l_{j^{\prime}}}{n_{j^{\prime}}}-(t^e_j+\frac{t^v_{j}+t^l_{j}}{n_j})$, and according to Eq.~\eqref{t1t2compare}:
  \begin{equation}
      T_1 > T_2.
  \end{equation}
  There exists a contradiction, and function $f_j$ and $f_{j^{\prime}}$ should exchange their places. Therefore, functions with a lower weight should be processed before functions with a higher weight. Therefore, Theorem~\ref{SSFSoptimal} holds.
 \end{proof}

\begin{algorithm}[t]
  \caption{\label{ESFFAlgorithm} Enhanced Shortest Function First}
    \If{A request $r_i$ arrives at the edge server}{
    Invoke Algorithm~\ref{FCP}\;
    }
    \If{A function instance $k^j_o$ finishes execution}{
    Invoke Algorithm~\ref{FRP}\;
    }
\end{algorithm}


\section{ESFF Algorithm}
Inspired by the optimal scheduling algorithm of SSFS, the Enhanced Shortest Function First (ESFF) algorithm is designed for online scheduling by considering multiple function instances and unknown future information. First, in ESFF, the average execution time $\overline{t^e_j}$, average cold start latency $\overline{t^l_j}$, and average eviction time $\overline{t^v_j}$ of function $f_j$ are computed according to the history. Then, the definition of function weight is modified to adapt to the SFS problem. 

As shown in Algorithm~\ref{ESFFAlgorithm}, ESFF consists of two sub-policies, invoked at different time: (a) Function Creation Policy (FCP), invoked when a new request $r_i$ of function $f_j$ arrives. FCP decides whether a new function instance should be initialized. (b) Function Replacement Policy (FRP), invoked when a request $r_i$ is finished by an instance of function $f_j$. FRP decides whether this instance should be replaced by other function instances. At other time, each function instance continuously processes the requests waiting in its queue.


\subsection{Function Creation Policy}
When a request $r_i$ of function $f_j$ arrives, FCP is invoked to judge whether a new instance of $f_j$ should be initialized. The total number of $f_j$'s requests waiting in its queue $q_j$ at $r_i$'s arrival time is denoted as $n^w_j$. 


\begin{algorithm}[t]
  \caption{\label{FCP} Function Creation Policy}
  \KwIn{$\overline{t^l_j}$, $\overline{t^v_j}$, $\overline{t^e_j}$, $n^w_j$, $r_i$,  $q_j$, $\mathbf{F}$, $\mathbf{K}^j$}
    \If{$n^w_j=0$ and $\sum_{k^j_o \in \mathbf{K}^j}{state(k^j_o)} > 0$}{
    Process $r_i$ in $f_j$'s idle instance\;
    }
    \Else{
    \If{$\sum_{f_j \in \mathbf{F}}{|\mathbf{K}^j|} < C$}{
    Compute $n^e_j$ based on Eq.~\eqref{afterinitializeremain}\;
    \If{$n^e_j > 0$}{
    Initialize an instance for function $f_j$\;
    }
    }
    \ElseIf{$\sum_{f_j \in \mathbf{F}}{|\mathbf{K}^j|} = C$}{
    Compute $\mathbf{S}$ based on Eqs.~\eqref{afterinitializeandevictremain} and~\eqref{largeS}\;
    \If{$|\mathbf{S}|> 0$}{
    $f_{j^{\prime}} = \argmax_{f_{j^{\prime}} \in \mathbf{S}}{\overline{t^e_{j^{\prime}}}}$\;
    Replace an idle instance of $f_{j^{\prime}}$ with $f_j$\;
    }
    }
    $q_j$.join($r_i$)\;
    }
\end{algorithm}


Algorithm~\ref{FCP} depicts the pseudocode of FCP. The key idea of FCP is that it makes scheduling decisions according to the current state of the system. If $n^w_j = 0$ (i.e., the queue $q_j$ is empty) and there exists an idle instance of the function $f_j$, the request $r_i$ is directly sent to an idle function instance for execution. In lines 4-7, if there is no idle function instance of $f_j$ and there are enough resources for initializing a new function instance of $f_j$, FCP prefers to initialize a new function instance with free resources than replace an idle function instance. It assesses whether the total response time can benefit from the function instance initialization, i.e., at least one request will be processed by the new function instance. The number of $f_j$'s waiting requests after the cold start latency $\overline{t^l_j}$ is estimated by:
  \begin{equation} \label{afterinitializeremain}
    \setlength{\belowdisplayskip}{2pt}
    \setlength{\abovedisplayskip}{2pt}
      n^e_j = n^w_j + 1 - \frac{\overline{t^l_j}|\mathbf{K}^j|}{\overline{t^e_j}}.
  \end{equation}
If $n^e_j > 0$, then it is likely that at least one request will be processed by the newly initialized function instance of $f_j$, so a new instance will be initialized. Otherwise, no function instance will be initialized. The request $r_i$ will join the queue $q_j$ in this and the following conditions.In lines 8-12, if there is no idle function instance of $f_j$ and the resources are insufficient to initialize such instance, it assesses whether the total response time can benefit from replacing an idle function instance. The number of $f_j$'s waiting requests after $f_{j^{\prime}}$'s eviction time and $f_j$'s cold start latency is estimated by:
  \begin{equation} \label{afterinitializeandevictremain}
      \setlength{\belowdisplayskip}{2pt}
    \setlength{\abovedisplayskip}{2pt}
      n^e_{j,j^{\prime}} = n^w_j + 1 - \frac{(\overline{t^l_j}+\overline{t^v_{j^{\prime}}})|\mathbf{K}^j|}{\overline{t^e_j}}.
  \end{equation}
  The candidate set of functions with enough eviction time that makes $n^e_{j,j^{\prime}} > 0$ is computed by:
  \begin{equation} \label{largeS}
  \setlength{\belowdisplayskip}{2pt}
    \setlength{\abovedisplayskip}{2pt}
      \mathbf{S}= \{f_{j^{\prime}}|n^e_{j,j^{\prime}} > 0 \,\, and \, \sum_{k^{j^{\prime}}_o \in \mathbf{K}^{j^{\prime}}}{state(k^{j^{\prime}}_o)} > 0\}.
  \end{equation}
  FCP chooses $f_{j^{\prime}} \in \mathbf{S}$ with the largest $\overline{t^e_{j^{\prime}}}$ and replaces an idle instance of $f_{j^{\prime}}$ with a new instance of $f_j$.
Otherwise, no function instance is initialized.



\begin{algorithm}[t]
  \caption{\label{FRP} Function Replacement Policy}
  \KwIn{$k^j_o$, $\overline{t^e_j}$, $\overline{t^v_j}$, $\mathbf{K}^j$, $\overline{t^l_j}$, $n^w_{j^{\prime}}$}  
    Compute $w_j$ based on Eq.~\eqref{initializedweight}\;
    $\mathbf{S} = \{f_{j^{\prime}}| n^w_{j^{\prime}}> 0\}$\;
    $f_x = f_j$, $w_x = w_j$\;
    \If{$|\mathbf{S}| < 1$}{
    \For{$j^{\prime} \in \mathbf{S}$}{
    Estimate $n^e_{j^{\prime},j}$ based on Eq.~\eqref{afterinitializeandevictremain}\;
    Compute $w_{j^{\prime}}$ based on Eq.~\eqref{wjprime}\;
    \If{$w_{j^{\prime}} < w_x$}{
    $w_x = w_{j^{\prime}}$, $f_x = f_{j^{\prime}}$\;
    }
    }
    }
    \If{$f_x \neq f_j$}{
    Replace $k^j_o$ with a new instance of function $f_x$\;
    }
    \ElseIf{$n^w_j>0$}{
    $k^j_o$ processes the first request in $q_j$\;}
    \Else{
    $state(k^j_o)=1$\;
    }
\end{algorithm}

\subsection{Function Replacement Policy}
When an instance $k^j_o$ of the function $f_j$ completes the execution of a request, FRP is invoked to decide whether $k^j_o$ should be replaced. Similar to the optimal scheduling algorithm of the SSFS, the function $f_j$ is associated with a weight to represent its scheduling urgency, defined as:
\begin{equation} \label{initializedweight}
\setlength{\belowdisplayskip}{2pt}
    \setlength{\abovedisplayskip}{2pt}
    w_j = \overline{t^e_j} + \frac{\overline{t^v_j}|\mathbf{K}^j|}{n^w_j}.
\end{equation}
The item $\overline{t^l_j}$ is removed from the numerator in Eq.~\eqref{initializedweight} for the function $f_j$ is already initialized.


 Algorithm~\ref{FRP} shows the pseudocode of FRP. The prime idea of FRP is trying to replace the existing function instance with a more urgent function based on the comparison of function weights. The scheduling decision made by FRP also depends on the current state of the system as follows. In lines 1-9, FRP finds the function that can benefit most from replacing the instance $k^j_o$. Since the future information of requests is unavailable, we only consider the current requests. The function's queue should have waiting requests. Otherwise, there is no need for initializing a new function instance. The candidate function set is $\mathbf{S} = \{f_{j^{\prime}}| n^w_{j^{\prime}}> 0\}$. Similar to Algorithm~\ref{FCP}, the number of $f_{j^{\prime}}$'s waiting requests, $n^e_{j^{\prime},j}$, after evicting $k^j_o$ and initializing a new instance of $f_{j^{\prime}}$ is estimated based on Eq.~\eqref{afterinitializeandevictremain} in line 6. In line 7, the weight of each function in the candidate function set is computed by:
\begin{equation} \label{wjprime}
\setlength{\belowdisplayskip}{2pt}
    \setlength{\abovedisplayskip}{2pt}
   w_{j^{\prime}} = \overline{t^e_{j^{\prime}}} + \frac{(\overline{t^l_{j^{\prime}}}+\overline{t^v_{j^{\prime}}})(|\mathbf{K}^j|+1)}{n^e_{j^{\prime},j}}.
\end{equation}
 Finally, the function $f_{j^{\prime}} \in \mathbf{S}$ with the smallest $w_{j^{\prime}}$ and $w_{j^{\prime}} \leq w_j$ is selected to replace the function instance $k^j_o$. In lines 12-13, if there is no function satisfying the above requirements, the function instance $k^j_o$ will not be replaced. If there is any request waiting in the queue $q_j$, $k^j_o$ processes the first request at the head of the queue. Otherwise, $k^j_o$ turns idle, $state(k^j_o) = 1$.

\begin{figure}[tbp]
    \setlength{\abovecaptionskip}{0cm} 
    \begin{center}
        \setlength{\belowcaptionskip}{0cm} 
      \includegraphics[width=0.65\linewidth]{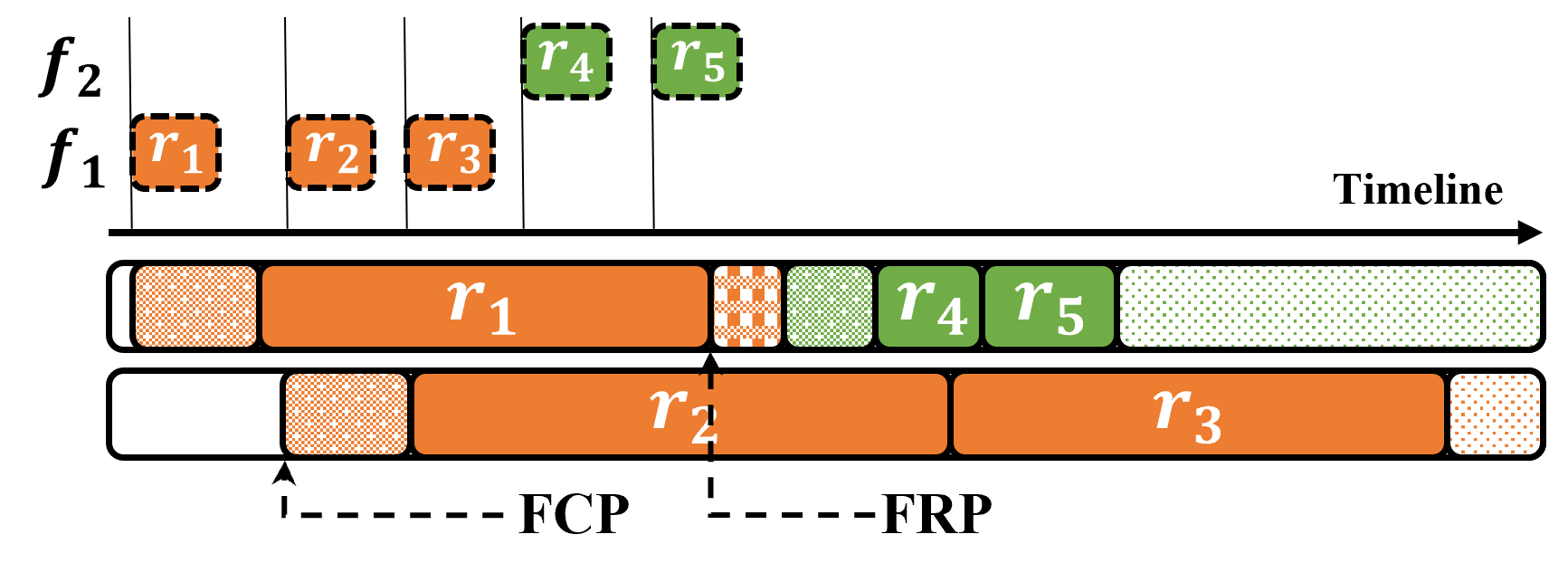}
      \caption{\label{schedulingexample} The scheduling example of ESFF. FCP initializes a new function instance for the request $r_2$, and replaces the function instance for request $r_4$.}
  \end{center}
  \vspace{-0.3cm}
  \end{figure}

\subsection{Scheduling Example}
Fig.~\ref{schedulingexample} is a scheduling example of ESFF. The edge server's capacity is two. Five requests of two functions arrive at the edge server at different times. Requests $r_1$, $r_2$ and $r_3$ are of function $f_1$ and requests $r_4$, $r_5$ are of function $f_2$. The request $r_1$ incurs a cold start at the beginning. When $r_2$ arrives, FCP decides to initialize a new instance of function $f_1$ to reduce $r_2$'s waiting time. At the arrival time of $r_3$, all function instances are busy and the computation resources are insufficient, so $r_3$ joins the queue $q_1$. When $r_1$ is finished, FRP decides to initialize a new instance of function $f_2$ to replace the instance of function $f_1$ since $w_2 < w_1$ and $n^w_{j^{\prime}} > 0$. 


\begin{figure*}[t]
    \setlength{\abovecaptionskip}{0cm} 
    \setlength{\belowcaptionskip}{0cm} 
    \centering
    \subfigure[Average Response Time]{
    \centering
    \includegraphics[width=0.31\textwidth]{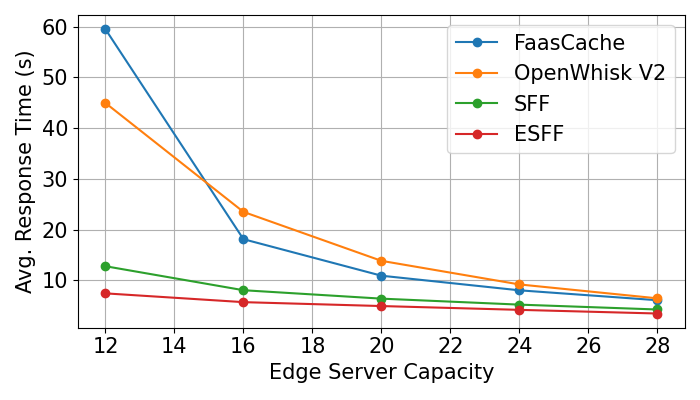}
    \label{exp_1_response_time}
    }%
    \subfigure[Average Slowdown]{
    \centering
    \includegraphics[width=0.31\textwidth]{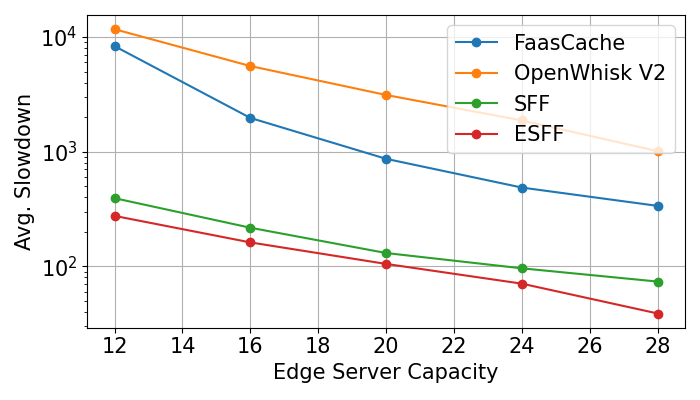}
    \label{exp_1_slowdown}
    }
    \subfigure[Average Cold Start Time]{
    \centering
    \includegraphics[width=0.31\textwidth]{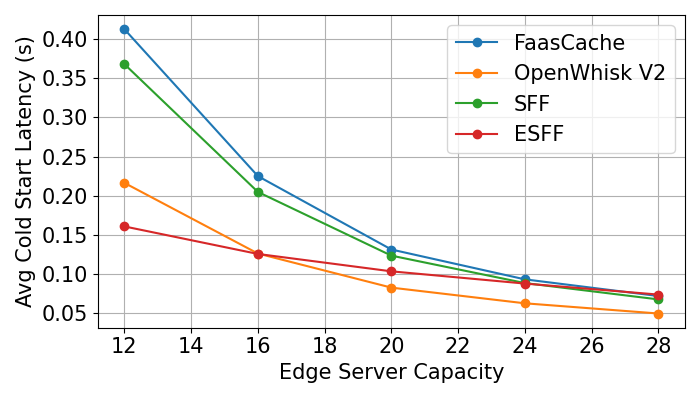}
    \label{exp_1_cold_start}
    }
    \centering
    \caption{\label{exp_1}Average response time, slowdown and cold start time under different edge server capacities.}
    \vspace{-0.3cm}
  \end{figure*}

  \begin{figure*}[t]
  \setlength{\abovecaptionskip}{0cm} 
    \setlength{\belowcaptionskip}{0cm} 
    \centering
    \subfigure[Average Response Time]{
    \centering
    \includegraphics[width=0.31\textwidth]{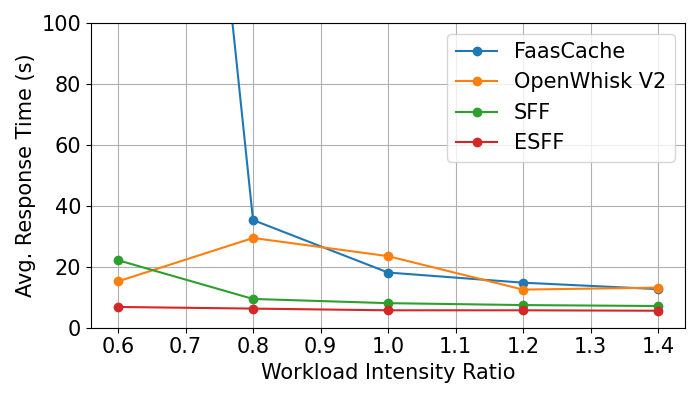}
    \label{exp_2_response_time}
    }%
    \subfigure[Average Slowdown]{
    \centering
    \includegraphics[width=0.31\textwidth]{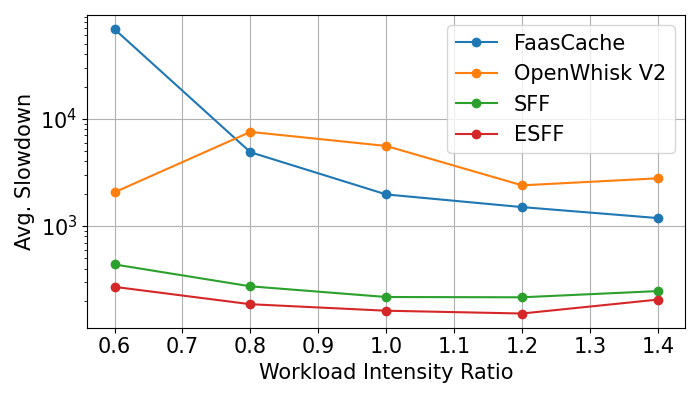}
    \label{exp_2_slowdown}
    }
    \subfigure[Average Cold Start Time]{
    \centering
    \includegraphics[width=0.31\textwidth]{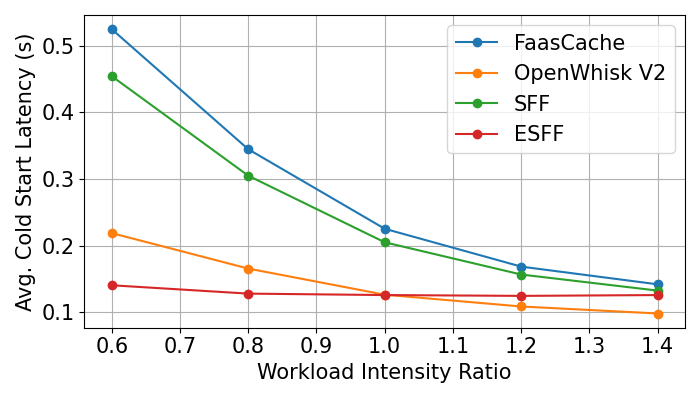}
    \label{exp_2_cold_start}
    }
    \centering
    \caption{\label{exp_2}Average response time, slow down and cold start time under different workload intensity ratios.}
    \vspace{-0.3cm}
  \end{figure*}

\section{Evaluation} 
\subsection{Setups}
 The ESFF algorithm, baseline algorithms, and the simulation environment are implemented in Python 3.6 on a desktop with an Intel Core i9-10900K 3.7 GHz CPU and 32GB RAM.

The serverless request traces are collected from a real cluster of Azure, containing $2.2\times 10^6$ requests during two weeks \cite{zhang2021faster}. The function, completion time, and execution time of each request are recorded. Limited by the measurement precision, some requests' execution time is recorded as 0, and is set to 1ms in the following experiments. We choose the first $6\times10^5$ requests as the simulation dataset. 

The default capacity of the resource-limited edge server is set to 16. Each function's cold start latency and eviction time are randomly chosen from [0.5,1.5] according to~\cite{yu2020characterizing}, since function details (e.g.,function codes and dependencies) are not included in the request traces \cite{zhang2021faster}.

We select three most related serverless function scheduling algorithms as baselines: (1) FaasCache \cite{fuerst2021faascache}. It schedules requests sequentially based on the arrival time. When there is no idle function instance for the current request, FaasCache tries to initialize a new instance or replace another instance. (2) OpenWhisk V2 \cite{OpenWhiskscheduler}. It makes requests of the same function wait in an individual queue. When the request at the queue head has waited for more than a fixed threshold (i.e., 100ms), a corresponding function instance is initialized. (3) Shortest Function First (SFF). The only difference between SFF and OpenWhisk is that SFF schedules requests sequentially based on the ascendant order of their average execution time.

The following two metrics are used to evaluate ESFF: (1) Average response time, $\frac{1}{|\mathbf{R}|}\sum_{r_i\in\mathbf{R}}{(t^c_i-t^a_i)}$. (2) Average slowdown, $\frac{1}{|\mathbf{R}|}\sum_{r_i\in\mathbf{R}}{\frac{t^c_i-t^a_i}{t^e_i}}$.

\subsection{Comparison with Baselines}
\textbf{Impact of edge server capacity.} In Fig.~\ref{exp_1}, ESFF is compared against baselines with respect to different edge server capacities. In Fig.~\ref{exp_1_response_time}, compared with the best baseline SFF, ESFF substantially reduces the average response time by 18\% to 40\%. When the edge server capacity increases, the function replacement times will be reduced for all scheduling algorithms. Therefore, the average cold start time is also reduced as shown in Fig.~\ref{exp_1_cold_start}. Besides, with more resources, the edge server can better handle the request bursts, and then request blocking problem will be mitigated. FaasCache and OpenWhisk V2 have very poor performance under very limited resources since they are unaware of massive cold starts and request blocking. Fig.~\ref{exp_1_slowdown} proves that by mitigating the request blocking problem, ESFF successfully reduces the response time of short requests and achieves the best fairness.

\begin{figure}[t]
    \setlength{\abovecaptionskip}{0cm} 
    \setlength{\belowcaptionskip}{0cm} 
    \centering
    \subfigure[CDF of Response Time]{
    \centering
    \includegraphics[width=0.23\textwidth]{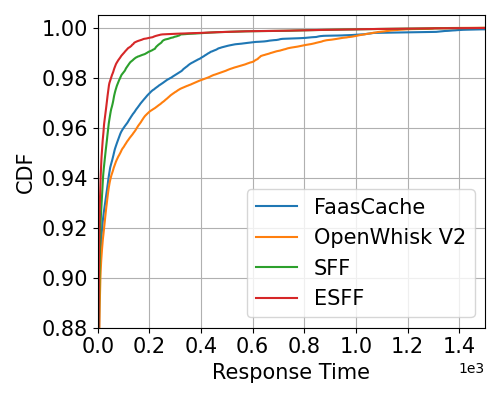}
    \label{cdf_response_time}
    }%
    \subfigure[CDF of Slowdown]{
    \centering
    \includegraphics[width=0.23\textwidth]{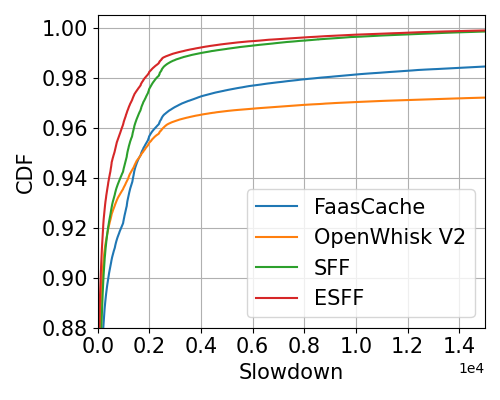}
    \label{cdf_slowdown}
    }
    \centering
    \caption{\label{cdf}CDF of response time and slowdown.}
    \vspace{-0.3cm}
  \end{figure}

  
\textbf{Impact of workload intensity.} In Fig.~\ref{exp_2}, the workload intensity ratio is defined to control the intervals between requests. In Fig.~\ref{exp_2_response_time}, ESFF achieves the lowest average response time under different workload intensity. With lower workload intensity ratios, the improvement of ESFF is more. In Fig.~\ref{exp_2_slowdown}, we observe that the average slowdown of ESFF increases as the workload intensity ratio increases from 1.0 to 1.4. The reason is that with longer intervals, the function instances of short requests will be more likely replaced, which brings more cold starts.

The cumulative distribution function (CDF) of response time and slowdown is shown in Fig.~\ref{cdf}. Compared with the baselines, the CDF curve of ESFF is always closer to the left, which means that the response time and slowdown of ESFF are consistently better. In particular, the P95 and P99 response time of ESFF are much better than other baselines. 


  \begin{figure}[tbp]
      \setlength{\abovecaptionskip}{0cm} 
    \setlength{\belowcaptionskip}{0cm} 
    \begin{center}
      \includegraphics[width=1\linewidth]{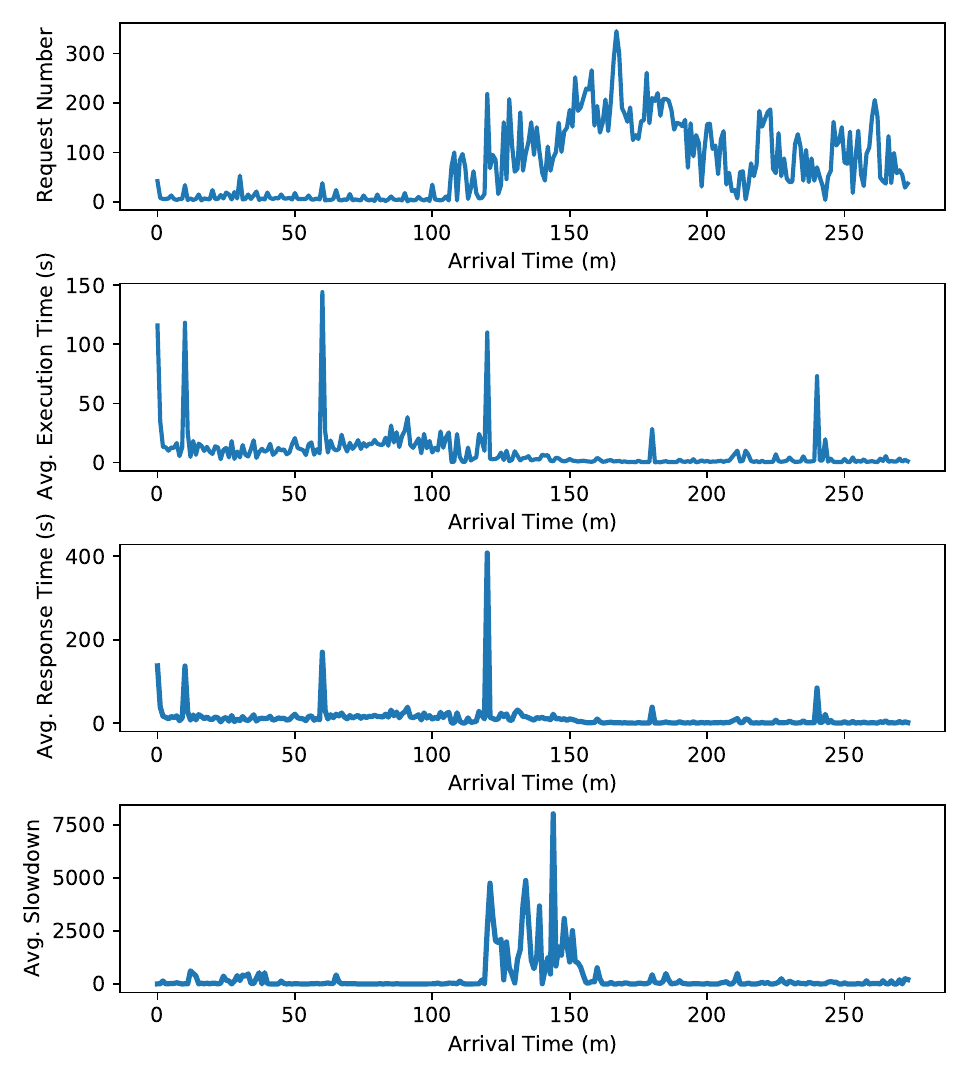}
      \caption{\label{deep_diving} Detailed scheduling results of ESFF over request arrival time. }
  \end{center}
  \vspace{-0.3cm}
  \end{figure}
  
\subsection{Further Analysis}
In Fig.~\ref{deep_diving}, we select 20000 continuous arriving requests to show detailed results of ESFF on each minute. Three points can be observed in this experiment: (1) It is obvious that with longer execution time, the response time also increases. (2) More request numbers further prolong the response time since the longer waiting time (i.e., short requests are blocked by long requests). (3) A large number of long requests not only affect the current requests but also largely slow down the upcoming requests. From these observations, we learn that the request burst (i.e., request number and request size) leads to significant response time. To deal with request bursts, offloading long requests to the powerful cloud will be a better choice. 


\section{Conclusion}
In this paper, we formulated the SFS problem for resource-limited edge computing. We proposed a polynomial-time optimal scheduling algorithm for a simplified offline form of SFS. Then, we designed an Enhanced Shortest Function First (ESFF) algorithm. To avoid wasteful cold starts, ESFF selectively decides the initialization of new function instances when requests arrive. To deal with dynamic workloads, ESFF judiciously replaces function instances based on function weights. Extensive simulations based on real-world traces show ESFF's substantial outperformance over existing baselines.

\bibliographystyle{IEEEtran}

\begin{thebibliography}{10}
\providecommand{\url}[1]{#1}
\csname url@samestyle\endcsname
\providecommand{\newblock}{\relax}
\providecommand{\bibinfo}[2]{#2}
\providecommand{\BIBentrySTDinterwordspacing}{\spaceskip=0pt\relax}
\providecommand{\BIBentryALTinterwordstretchfactor}{4}
\providecommand{\BIBentryALTinterwordspacing}{\spaceskip=\fontdimen2\font plus
\BIBentryALTinterwordstretchfactor\fontdimen3\font minus \fontdimen4\font\relax}
\providecommand{\BIBforeignlanguage}[2]{{%
\expandafter\ifx\csname l@#1\endcsname\relax
\typeout{** WARNING: IEEEtran.bst: No hyphenation pattern has been}%
\typeout{** loaded for the language `#1'. Using the pattern for}%
\typeout{** the default language instead.}%
\else
\language=\csname l@#1\endcsname
\fi
#2}}
\providecommand{\BIBdecl}{\relax}
\BIBdecl

\bibitem{shi2016edge}
W.~Shi, J.~Cao, Q.~Zhang, Y.~Li, and L.~Xu, ``Edge computing: Vision and challenges,'' \emph{IEEE internet of things journal}, vol.~3, no.~5, pp. 637--646, 2016.

\bibitem{ouyang2018follow}
T.~Ouyang, Z.~Zhou, and X.~Chen, ``Follow me at the edge: Mobility-aware dynamic service placement for mobile edge computing,'' \emph{IEEE Journal on Selected Areas in Communications}, vol.~36, no.~10, pp. 2333--2345, 2018.

\bibitem{castro2019rise}
P.~Castro, V.~Ishakian, V.~Muthusamy, and A.~Slominski, ``The rise of serverless computing,'' \emph{Communications of the ACM}, vol.~62, no.~12, pp. 44--54, 2019.

\bibitem{aslanpour2021serverless}
M.~S. Aslanpour, A.~N. Toosi, C.~Cicconetti, B.~Javadi, P.~Sbarski, D.~Taibi, M.~Assuncao, S.~S. Gill, R.~Gaire, and S.~Dustdar, ``Serverless edge computing: vision and challenges,'' in \emph{Proceedings of the 2021 Australasian Computer Science Week Multiconference (ACSW)}, 2021, pp. 1--10.

\bibitem{wang2018peeking}
L.~Wang, M.~Li, Y.~Zhang, T.~Ristenpart, and M.~Swift, ``Peeking behind the curtains of serverless platforms,'' in \emph{Proceedings of the 2018 USENIX Annual Technical Conference (USENIX ATC)}, 2018, pp. 133--146.

\bibitem{openwhisk}
\BIBentryALTinterwordspacing
Apache openwhisk. [Online]. Available: \url{https://openwhisk.apache.org}
\BIBentrySTDinterwordspacing

\bibitem{awslambda}
\BIBentryALTinterwordspacing
Aws lambda. [Online]. Available: \url{https://aws.amazon.com/lambda/}
\BIBentrySTDinterwordspacing

\bibitem{Azure}
\BIBentryALTinterwordspacing
Microsoft azure serverless functions. [Online]. Available: \url{https://azure.microsoft.com/en-us/services/functions}
\BIBentrySTDinterwordspacing

\bibitem{gunasekaran2020fifer}
J.~R. Gunasekaran, P.~Thinakaran, N.~C. Nachiappan, M.~T. Kandemir, and C.~R. Das, ``Fifer: Tackling resource underutilization in the serverless era,'' in \emph{Proceedings of the 21st International Middleware Conference (Middleware)}, 2020, pp. 280--295.

\bibitem{wu2022container}
S.~Wu, Z.~Tao, H.~Fan, Z.~Huang, X.~Zhang, H.~Jin, C.~Yu, and C.~Cao, ``Container lifecycle-aware scheduling for serverless computing,'' \emph{Software: Practice and Experience}, vol.~52, no.~2, pp. 337--352, 2022.

\bibitem{pan2022retention}
L.~Pan, L.~Wang, S.~Chen, and F.~Liu, ``Retention-aware container caching for serverless edge computing,'' in \emph{Proceedings of the 41st IEEE Conference on Computer Communications (INFOCOM)}, 2022, pp. 1069--1078.

\bibitem{hall2019execution}
A.~Hall and U.~Ramachandran, ``An execution model for serverless functions at the edge,'' in \emph{Proceedings of the International Conference on Internet of Things Design and Implementation (IoTDI)}, 2019, pp. 225--236.

\bibitem{das2020performance}
A.~Das, S.~Imai, S.~Patterson, and M.~P. Wittie, ``Performance optimization for edge-cloud serverless platforms via dynamic task placement,'' in \emph{Proceedings of the 2020 20th IEEE/ACM International Symposium on Cluster, Cloud and Internet Computing (CCGRID)}, 2020, pp. 41--50.

\bibitem{cicconetti2021faas}
C.~Cicconetti, M.~Conti, and A.~Passarella, ``Faas execution models for edge applications,'' \emph{arXiv preprint arXiv:2111.06595}, 2021.

\bibitem{wang2021lass}
B.~Wang, A.~Ali-Eldin, and P.~Shenoy, ``Lass: running latency sensitive serverless computations at the edge,'' in \emph{Proceedings of the 30th International Symposium on High-Performance Parallel and Distributed Computing (HPDC)}, 2021, pp. 239--251.

\bibitem{bermbach2022auctionwhisk}
D.~Bermbach, J.~Bader, J.~Hasenburg, T.~Pfandzelter, and L.~Thamsen, ``Auctionwhisk: Using an auction-inspired approach for function placement in serverless fog platforms,'' \emph{Software: Practice and Experience}, vol.~52, no.~5, pp. 1143--1169, 2022.

\bibitem{OpenWhiskscheduler}
\BIBentryALTinterwordspacing
Openwhisk future architecture. [Online]. Available: \url{https://cwiki.apache.org/confluence/display/OPENWHISK/OpenWhisk+future+architecture}
\BIBentrySTDinterwordspacing

\bibitem{fuerst2021faascache}
A.~Fuerst and P.~Sharma, ``Faascache: keeping serverless computing alive with greedy-dual caching,'' in \emph{Proceedings of the 26th ACM International Conference on Architectural Support for Programming Languages and Operating Systems (ASPLOS)}, 2021, pp. 386--400.

\bibitem{9796969}
C.~Zhang, H.~Tan, G.~Li, Z.~Han, S.~H.-C. Jiang, and X.-Y. Li, ``Online file caching in latency-sensitive systems with delayed hits and bypassing,'' in \emph{Proceedings of the 41st IEEE Conference on Computer Communications (INFOCOM)}, 2022, pp. 1059--1068.

\bibitem{zhang2021faster}
Y.~Zhang, {\'I}.~Goiri, G.~I. Chaudhry, R.~Fonseca, S.~Elnikety, C.~Delimitrou, and R.~Bianchini, ``Faster and cheaper serverless computing on harvested resources,'' in \emph{Proceedings of the ACM SIGOPS 28th Symposium on Operating Systems Principles (SOSP)}, 2021, pp. 724--739.

\bibitem{varghese2016challenges}
B.~Varghese, N.~Wang, S.~Barbhuiya, P.~Kilpatrick, and D.~S. Nikolopoulos, ``Challenges and opportunities in edge computing,'' in \emph{Proceedings of the 2016 IEEE International Conference on Smart Cloud (SmartCloud)}, 2016, pp. 20--26.

\bibitem{8057116}
H.~Tan, Z.~Han, X.-Y. Li, and F.~C. Lau, ``Online job dispatching and scheduling in edge-clouds,'' in \emph{Proceedings of the 36st IEEE Conference on Computer Communications (INFOCOM)}, 2017, pp. 1--9.

\bibitem{liaee1997scheduling}
M.~M. Liaee and H.~Emmons, ``Scheduling families of jobs with setup times,'' \emph{International Journal of Production Economics}, vol.~51, no.~3, pp. 165--176, 1997.

\bibitem{yu2020characterizing}
T.~Yu, Q.~Liu, D.~Du, Y.~Xia, B.~Zang, Z.~Lu, P.~Yang, C.~Qin, and H.~Chen, ``Characterizing serverless platforms with serverlessbench,'' in \emph{Proceedings of the 11th ACM Symposium on Cloud Computing (SoCC)}, 2020, pp. 30--44.

\end{thebibliography}

\end{document}